\documentclass[11pt,onecolumn]{IEEEtran}

\usepackage{amsmath,amssymb,amsthm,mathtools}
\usepackage{cite}
\usepackage{enumitem}
\usepackage{hyperref}
\hypersetup{hidelinks}
\usepackage{microtype}
\usepackage{xcolor}
\usepackage{cite}

\makeatletter
\renewcommand{\citepunct}{,\penalty\@m\ }
\makeatother

\newtheorem{theorem}{Theorem}[section]
\newtheorem{lemma}[theorem]{Lemma}

\newtheorem{corollary}[theorem]{Corollary}
\newtheorem{definition}[theorem]{Definition}
\newtheorem{remark}[theorem]{Remark}

\newcommand{\BD}{\mathcal{BD}}
\newcommand{\BI}{\mathcal{BI}}
\newcommand{\DVT}{\mathrm{DVT}}
\newcommand{\Diff}{\mathrm{Diff}}
\newcommand{\Syn}{\mathrm{Syn}}

\title{New bounds for covering codes under insertions or deletions}

\author{Chengfei Xie, Yubo Sun, and Gennian Ge%
\thanks{This research was supported by the National Key Research and Development Program of China under Grant 2025YFC3409900, the National Natural Science Foundation of China under Grant 12231014, and Beijing Scholars Program. The work of C. Xie was supported by the National Natural Science Foundation of China under Grant No. 12401440.}
\thanks{C. Xie is with the Institute of Mathematics and Interdisciplinary Sciences, Xidian University, Xi'an 710126, China (e-mail: cfxie@cnu.edu.cn).}%
\thanks{Y. Sun and G. Ge are with the School of Mathematical Sciences, Capital Normal University, Beijing 100048, China (e-mail: 2200502135@cnu.edu.cn; gnge@zju.edu.cn).}}

\begin{document}
\maketitle

\begin{abstract}
Covering codes for insertions and deletions arise naturally in the study of synchronization errors and differ substantially from their classical counterparts in the Hamming metric. In this paper, we study covering codes under insertion and deletion operations. We first show that, in contrast to the equivalence between insertion and deletion correction, insertion covering and deletion covering are not equivalent. We then develop bounds and constructions for insertion and deletion covering codes, with particular emphasis on the large-alphabet regime. For insertion covering codes, we extend a recent combinatorial approach for single insertions and establish a new lower bound for arbitrary fixed insertion radius. For deletion covering codes, we relate the problem to hypergraph covering and prove that the elementary counting lower bound is asymptotically tight when the alphabet size tends to infinity. We further provide a construction of asymptotically optimal non-binary single-deletion covering codes by using differential Varshamov--Tenengolts (VT) codes together with a completion argument. In addition, we study covering codes for burst deletions. We prove that binary differential VT codes are not only capable of correcting two-burst deletions but also have the corresponding covering property, and hence form binary perfect codes for two-burst deletions. Finally, we extend this construction to non-binary alphabets and obtain explicit $q$-ary two-burst-deletion covering codes.
\end{abstract}


\section{Introduction}

Covering codes are a fundamental object in coding theory and discrete mathematics, with a wide range of applications including data compression \cite{cohen1997covering}, combinatorial games such as football pools \cite{hamalainen1995football}, circuit complexity \cite{smolensky1993representations}, lattice problems \cite{micciancio2004almost}, and approximate nearest-neighbor search \cite{pagh2016locality}. Much of the classical theory has been developed for the Hamming metric, where errors are modeled by symbol substitutions \cite{cohen1997covering}. In many modern applications involving strings, however, insertions and deletions provide a more natural error model. This is particularly relevant for textual and biological data, where synchronization errors occur naturally and are captured by the Levenshtein distance \cite{Lenz-21-IT}.

Covering problems under insertions and deletions differ substantially from their Hamming-metric counterparts. In the Hamming setting, both the codewords and the words to be covered lie in the same ambient space. By contrast, insertions and deletions change the word length, so the codewords and the covered words belong to spaces of different dimensions. This asymmetry, together with the irregular behavior of deletion balls, leads to new combinatorial difficulties that are absent in the substitution setting \cite{cohen1997covering,Lenz-21-IT}.


Let $[a,b]$ denote the set of integers $a,a+1,\ldots,b$, and write $[N]=[1,N]$. Let $\Sigma_q=[0,q-1]$ be the $q$-ary alphabet. For words $x$ and $y$, write $x\preceq y$ if $x$ is a subsequence of $y$.
For a positive integer $R$ and a word $x\in\Sigma_q^n$, we say that a word $y\in\Sigma_q^{n+R}$ is \emph{covered} by $x$ under $R$ insertions if $y$ can be obtained from $x$ by inserting exactly $R$ symbols, or equivalently $x\preceq y$. The set of all such words is called the \emph{$R$-insertion ball} centered at $x$, and is denoted by $\mathcal I_R(x)$. A code $C\subseteq\Sigma_q^n$ is an \emph{$R$-insertion covering code} if
\[
   \bigcup_{x\in C} \mathcal I_R(x)=\Sigma_q^{n+R}.
\]
We denote by $I(q,n,R)$ the minimum possible size of an $R$-insertion covering code. Since the size of an $R$-insertion ball is independent of its center and is given by
\begin{equation}\label{eq:insball}
 |\mathcal I_R(x)|=\sum_{i=0}^R \binom{n+R}{i}(q-1)^i,
\end{equation}
the standard sphere-covering argument gives
\begin{equation}\label{eq:insertion-sphere-bound}
 I(q,n,R)\ge
 \frac{q^{n+R}}{\sum_{i=0}^R \binom{n+R}{i}(q-1)^i}.
\end{equation}

For deletions, a word $y\in\Sigma_q^{n-R}$ is said to be \emph{covered} by $x\in\Sigma_q^n$ under $R$ deletions if $y$ can be obtained from $x$ by deleting exactly $R$ symbols, or equivalently $y\preceq x$. The set of all such words is called the \emph{$R$-deletion ball} centered at $x$, and is denoted by $\mathcal D_R(x)$. A code $C\subseteq\Sigma_q^n$ is an \emph{$R$-deletion covering code} if
\[
   \bigcup_{x\in C} \mathcal D_R(x)=\Sigma_q^{n-R}.
\]
Let $D(q,n,R)$ denote the minimum size of such a code. Since every $R$-deletion ball has size at most $\binom nR$, the elementary counting bound gives
\begin{equation}\label{eq:deletion-counting-bound}
   D(q,n,R)\ge \frac{q^{n-R}}{\binom nR}.
\end{equation}

\subsection{Known results}


Covering codes under insertions and deletions were systematically studied by Lenz, Rashtchian, Siegel, and Yaakobi \cite{Lenz-21-IT}. For insertions, they established sphere-covering lower bounds and constructed covering codes by combining random coding with an inductive argument. In particular, they proved that
\[
\frac{q^{n+R}}{\sum_{i=0}^R\binom{n+R}{i}(q-1)^i}
\le I(q,n,R)
\le
\mu_I(R)\frac{q^{n+R}}{\sum_{i=0}^R\binom{n+R}{i}(q-1)^i},
\]
where $\mu_I(1)\le 7$ and, for $R\ge2$,
\[
  \mu_I(R)\le e\bigl(R\log R+\sqrt{2R\log R}+1\bigr)\mu_I(1).
\]
Thus, for fixed $R$, their upper and lower bounds differ by at most a factor of order $O(R\log R)$. The single-insertion case has recently been further refined by Pikhurko, Verbitsky, and Zhukovskii \cite{MR4910621}. By relating single-insertion covering codes to Turan densities in extremal combinatorics, they improved the bounds in the large-alphabet regime. In particular, for fixed word length $n$ and $q\to\infty$, their results imply
\[
  (1-o(1))\frac{q^n}{n}
  \le I(q,n,1)
  \le (1+o(1))\,4.911\,\frac{q^n}{n+1}.
\]

For deletion covering codes, the situation is more delicate because deletion balls are not regular: their sizes depend on structural properties of the center word, such as the number of runs. In the binary case, the classical Varshamov--Tenengolts (VT) codes provide perfect single-deletion-correcting codes \cite{Levenshtein1992}, and hence also yield binary single-deletion covering codes of asymptotically optimal size. Consequently, as $n\to\infty$,
\[
  D(2,n,1)=(1+o(1))\frac{2^n}{n+1}.
\]
For non-binary alphabets, known single-deletion-correcting codes are generally not perfect and therefore do not directly yield covering codes. Using the binary VT construction as a building block, Lenz, Rashtchian, Siegel, and Yaakobi \cite{Lenz-21-IT} constructed $q$-ary single-deletion covering codes of size
\[
   \frac{q^n}{(n+1)\lceil q/2\rceil}.
\]
More generally, for fixed $q$ and $R$ as $n\to\infty$, they proved that
\[
  \frac{q^nR!}{n^R(q-1)^R}
  \le D(q,n,R)
  \le \mu_D(R)\frac{q^nR!}{n^R(q-1)^R},
\]
where $\mu_D(1)\le (q-1)/\lceil q/2\rceil$ and, for $R\ge2$,
\[
\mu_D(R)\le e\bigl(R\log R+\sqrt{2R\log R}+1\bigr)\mu_D(1).
\]

\subsection{Our contributions}


In this paper, we study insertion and deletion covering codes in the large-alphabet regime. Our main contributions can be summarized as follows.
\begin{itemize}[leftmargin=2em]
\item \emph{Non-equivalence of insertion and deletion covering.} Covering codes can be viewed as a natural counterpart of error-correcting codes. For insertion and deletion errors, it is well known that a code can correct $R$ deletions if and only if it can correct $R$ insertions \cite{Levenshtein-66}. It is therefore natural to ask whether an analogous equivalence holds for insertion-covering and deletion-covering codes. We give a negative answer to this question. More precisely, we construct examples showing that an $R$-deletion covering code need not be an $R$-insertion covering code, and conversely, an $R$-insertion covering code need not be an $R$-deletion covering code.
\item \emph{Improved lower bounds for $R$-insertion covering codes.}  We generalize the combinatorial approach of Pikhurko, Verbitsky, and Zhukovskii \cite{MR4910621} and prove that
\begin{equation}\label{eq:intro-lower}
 I(q,n,R)\ge \frac{q^n}{\binom{n+R}{R}-1}.
\end{equation}
This strengthens the standard sphere-covering lower bound in the large-alphabet regime and extends the phenomenon observed for single insertions to arbitrary fixed insertion radius.
\item \emph{Asymptotically tight upper bounds for $R$-deletion covering codes over large alphabets.}  By relating deletion covering codes to hypergraph covering problems, we prove that, for fixed $n$ and $R$ and sufficiently large $q$,
\begin{equation}\label{eq:intro-deletion-upper}
 D(q,n,R)\le (1+o(1))\frac{q^{n-R}}{\binom nR}.
\end{equation}
Together with \eqref{eq:deletion-counting-bound}, this determines the asymptotic value of $D(q,n,R)$ in the large-alphabet regime.
In particular, for $R = 1$, we construct nearly optimal non-binary single-deletion covering codes by
starting from a non-binary single-deletion-correcting code and adding a carefully chosen set of codewords that
covers the remaining uncovered words. A pigeonhole argument shows the existence of such a construction with
asymptotically optimal size.
\item \emph{Explicit constructions for two-burst-deletion covering codes.}
We prove that binary differential VT codes, which are known to correct two-burst deletions, also satisfy the corresponding covering property. Consequently, they form binary perfect codes for two-burst deletions. We further extend this construction to non-binary alphabets by adapting the parity-and-checksum lifting method of Lenz, Rashtchian, Siegel, and Yaakobi~\cite{Lenz-21-IT}.
\end{itemize}

\subsection{Organization}

The remainder of the paper is organized as follows. Section~\ref{sec:prelim} introduces the notation and preliminary results used throughout the paper.
Section~\ref{sec:noneq} proves the non-equivalence of insertion-covering and deletion-covering codes.
Section~\ref{sec:insertion-lower} proves the lower bound for insertion covering codes. Section~\ref{sec:deletion-upper} studies deletion covering codes and presents bounds and constructions. Section~\ref{sec:burst} provides constructions for two-burst-deletion covering codes. Finally, Section~\ref{sec:conclusion} concludes the paper with several open problems.

\section{Preliminaries}\label{sec:prelim}

In this section we collect the combinatorial tools and coding-theoretic ingredients used in the proofs.
Let $C_1,C_2,\ldots,C_k$ be finite subsets of a space $X$. For two indices $i<j$, write $C_{i,j}=C_i\cap C_j$. By Bonferroni's inequality,
\[
 \left|\bigcup_{i=1}^k C_i\right|
 \ge \sum_{i=1}^k |C_i|-\sum_{1\le i<j\le k}|C_{i,j}|.
\]
Conversely, we shall use the following inverse Bonferroni inequality (see, e.g., \cite[Lemma 8]{MR4910621}).

\begin{lemma}\label{lem:inverse-bonferroni}
Let $k\ge2$ be an integer. Let $C_1,C_2,\ldots,C_k$ be finite subsets of a space $X$ and let $G$ be a tree with vertex set $\{1,2,\ldots,k\}$ and edge set $E$. For $e=\{i,j\}\in E$, define $C_e=C_i\cap C_j$. Then
\[
 \left|\bigcup_{i=1}^k C_i\right|
 \le \sum_{i=1}^k |C_i|-\sum_{e\in E}|C_e|.
\]
\end{lemma}

\begin{proof}
We use induction on $k$. When $k=2$, the assertion follows immediately from the inclusion-exclusion principle. For the induction step, choose a leaf of $G$, say $k$, and without loss of generality let its unique neighbor be $k-1$. Let $G'$ be the tree induced by $\{1,\ldots,k-1\}$ and let $E'$ be its edge set. Then we have
\begin{equation}\label{yige}
\begin{split}
\sum_{e\in E}|C_{e}|=&\left|C_k\bigcap C_{k-1}\right|+\sum_{e\in E'}|C_{e}|\\
\overset{(a)}\leq&\left|C_k\bigcap C_{k-1}\right|+\sum_{i=1}^{k-1}|C_i|-\left|\bigcup_{i=1}^{k-1}C_i\right|,
\end{split}
\end{equation}
where in $(a)$ we use the inductive hypothesis.
On the other hand,
\begin{equation}\label{erge}
 \begin{split}
  \left|\bigcup_{i=1}^kC_i\right|=&|C_k|+\left|\bigcup_{i=1}^{k-1}C_i\right|-\left|C_k\bigcap\left(\bigcup_{i=1}^{k-1}C_i\right)\right|\\
  \leq&|C_k|+\left|\bigcup_{i=1}^{k-1}C_i\right|-\left|C_k\bigcap C_{k-1}\right|.
\end{split}
\end{equation}
Combining inequalities (\ref{yige}) and (\ref{erge}) gives the desired inequality.
This completes the proof.
\end{proof}

The next ingredient is a covering theorem for almost regular hypergraphs, which will be used to construct deletion covering codes in the large-alphabet regime.

For an $r$-uniform hypergraph $\mathcal H=(V,E)$ and a vertex $x\in V$, let $d(x)$ denote the \emph{degree} of $x$ in $\mathcal H$, i.e., the number of hyperedges containing $x$.
For two vertices $x,y\in V$, let $d(x,y)$ denote the \emph{codegree} of $x$ and $y$, namely, the number of hyperedges containing both $x$ and $y$.
A \emph{cover} of $\mathcal H$ is a family of hyperedges whose union contains all vertices of $V$.
For a real number $\delta>0$, we write $a=1\pm\delta$ to mean $1-\delta\le a\le 1+\delta$.
The following theorem is due to Pippenger (unpublished). For a proof, see, e.g., \cite[Theorem~4.7.1]{MR3524748}.

\begin{theorem}\label{thm:pippenger}
For every integer $r\ge2$ and reals $k\ge1$ and $a>0$, there are $\gamma=\gamma(r,k,a)>0$ and $d_0=d_0(r,k,a)$ such that for every $n\geq D\geq d_0$, the following holds. Every $r$-uniform hypergraph $\mathcal{H} = (V, E)$ on a set $V$ of $n$ vertices in which all vertices have positive degrees and which satisfies the following conditions:
\begin{enumerate}[label=(\roman*),leftmargin=2em]
  \item For all vertices $x \in V$ but at most $\gamma n$ of them, $d(x) = (1 \pm \gamma)D$.
  \item For all $x \in V$, $d(x) < kD$.
  \item For any two distinct $x, y \in V$, $d(x, y) < \gamma D$.
\end{enumerate}
contains a cover of at most $(1+a)|V|/r$ hyperedges.
\end{theorem}

We also need a small amount of notation for burst insertions and deletions.
For a sequence $x$, let $\BI_b(x)$ be the \emph{$b$-burst-insertion ball} centered at $x$, consisting of all sequences obtainable from $x$ after a burst of exactly $b$ insertions. Similarly, let $\BD_b(x)$ be the \emph{$b$-burst-deletion ball} centered at $x$.

\begin{lemma}\cite[Theorem 1]{Sun-25-IT}\label{lem:burst-insertion-size}
For any $q\ge2$ and $x\in\Sigma_q^n$,
\[
   |\BI_b(x)|=q^{b-1}\bigl((q-1)(n+1)+1\bigr).
\]
\end{lemma}

Finally, we recall differential VT codes, which provide the main explicit ingredients in the deletion-covering constructions below. For $x=(x_1,x_2,\ldots,x_n)\in\Sigma_q^n$, define $\Diff(x)\in\Sigma_q^n$ by
\[
\Diff(x)_i=
\begin{cases}
 x_{i+1}-x_i \pmod q, & 1\le i\le n-1,\\
 x_n, & i=n.
\end{cases}
\]
For $u=(u_1,\ldots,u_n)\in\Sigma_q^n$, define its VT syndrome by
\[
   \Syn(u)=\sum_{i=1}^n i u_i.
\]
For any $a\in[0,qn-1]$, the \emph{$q$-ary differential VT code} is
\[
\DVT_a(n,q)=\{x\in\Sigma_q^n:\ \Syn(\Diff(x))\equiv a\pmod{qn}\}.
\]
The following facts are known.

\begin{lemma}\cite[Theorem 3]{Nguyen-24-IT}\label{lem:dvt-single-del}
For any $q\ge2$ and $a\in[0,qn-1]$, the code $\DVT_a(n,q)$ is a $q$-ary single-deletion correcting code.
\end{lemma}

\begin{lemma}\cite[Theorem 6]{Wang-24-IT}\label{lem:dvt-binary-burst}
For any $a\in[0,2n-1]$, the code $\DVT_a(n,2)$ is a binary two-burst-deletion correcting code.
\end{lemma}

\begin{lemma}\cite[Corollary 3]{Babu-25-ISIT}\label{lem:dvt-size-div}
If each prime divisor of $n$ is also a divisor of $q$, then, for any $a\in[0,nq-1]$,
\[
   |\DVT_a(n,q)|=\frac{q^{n-1}}{n}.
\]
\end{lemma}

\begin{lemma}\cite[Corollary 7]{Babu-25-ISIT}\label{lem:dvt-size-prime}
Suppose that $n$ is a prime that does not divide $q$. Then, for any $a\in[0,nq-1]$,
\[
 |\DVT_a(n,q)|=
 \begin{cases}
  \frac1n(q^{n-1}-1)+1, & n\mid a,\\
  \frac1n(q^{n-1}-1), & n\nmid a.
 \end{cases}
\]
\end{lemma}
\section{Non-equivalence of Insertion-Covering and Deletion-Covering Codes}\label{sec:noneq}

We begin with a structural distinction between correction and covering. This distinction explains why insertion-covering and deletion-covering codes should not be expected to behave symmetrically.

For error-correcting codes in the Levenshtein metric, correcting a fixed number of deletions is equivalent to correcting the same number of insertions.
Indeed, the deletion correction and insertion correction are governed by the same longest-common-subsequence condition.

For covering codes, the situation is different. An $R$-insertion-covering code $C\subseteq\Sigma_q^n$ covers the longer space $\Sigma_q^{n+R}$ by asking that every longer word contains a codeword as a subsequence, whereas an $R$-deletion-covering code covers the shorter space $\Sigma_q^{n-R}$ by asking that every shorter word is a subsequence of some codeword. Thus the two notions are inherently asymmetric. We make this asymmetry precise by showing that neither implication holds.
For a symbol $a\in\Sigma_q$, let $a^n$ denote the length-$n$ word consisting entirely of $a$.
For two words $x$ and $y$, let $xy$ denote their concatenation.

\begin{theorem}\label{thm:del-not-ins}
For every $q\ge 2$, every $R\ge 1$, and every $n>R$, let $C'\subseteq\Sigma_q^n$ be an $R$-deletion covering code.
Then, for any two distinct symbols $a,b\in\Sigma_q$, the code
\[
   C:=\bigl(C'\cup \{a^{n-R}b^R\}\bigr)\setminus \{a^n\}
\]
is an $R$-deletion covering code but is not an $R$-insertion covering code.
\end{theorem}

\begin{proof}
We first prove that $C$ is still an $R$-deletion covering code.
Let $y\in\Sigma_q^{n-R}$. Since $C'$ is an $R$-deletion covering code, there exists $c'\in C'$ such that $y\preceq c'$.
If $c'\neq a^n$, then $c'\in C$, and hence $y$ is covered by $C$.
It remains to consider the case $c'=a^n$. In this case, every length-$(n-R)$ subsequence of $a^n$ is equal to $a^{n-R}$, so necessarily
\[
   y=a^{n-R}.
\]
But
\[
   a^{n-R}\preceq a^{n-R}b^R,
\]
and the word $a^{n-R}b^R$ belongs to $C$. Therefore $y$ is also covered by $C$. This proves that $C$ is an $R$-deletion covering code.

On the other hand, $C$ is not an $R$-insertion covering code. Indeed, consider the word
\[
   z=a^{n+R}\in\Sigma_q^{n+R}.
\]
The only length-$n$ subsequence of $z$ is $a^n$. Since $a^n\notin C$, there is no codeword of $C$ that is a subsequence of $z$. Hence $z$ is not covered by $C$ under $R$ insertions, and so $C$ is not an $R$-insertion covering code.
\end{proof}

The reverse implication is less immediate. We prove it by constructing an explicit binary code family. The key idea is to forbid a carefully chosen short subsequence. This destroys deletion covering, while an extremal subsequence argument shows that insertion covering still holds.

\begin{theorem}\label{thm:main-noneq}
For every integer $R\ge 1$, let
\[
   C_R=\{c\in\{0,1\}^{R+4}: 0110\not\preceq c\}.
\]
Then $C_R$ is an $R$-insertion-covering code but is not an $R$-deletion-covering code.
\end{theorem}

\begin{proof}
We first show that $C_R$ is not an $R$-deletion-covering code.
Suppose that $C_R$ is an $R$-deletion-covering code. The word $u=0110$ must be covered by some codeword $c\in C_R$.
This would require $u\preceq c$, contradicting the definition
of $C_R$.

It remains to prove that $C_R$ is an $R$-insertion-covering
code. For any $z\in\{0,1\}^{(R+4)+R}=\{0,1\}^{2R+4}$, we shall show that there exists a set
$E\subseteq[2R+4]$ of size $R$ such that $c=z_{[2R+4]\setminus E}$ does not contain $0110$ as a subsequence. This will imply that $c\in C_R$ and $c\preceq z$.

Suppose, for contradiction, that every subsequence obtained
from $z$ by deleting exactly $R$ coordinates still contains
$0110$ as a subsequence.
Let $N_0$ and $N_1$ denote the numbers of zeros and ones in
$z$, respectively. We first claim that
\[
   N_0\ge R+2
   \qquad\text{and}\qquad
   N_1\ge R+2.
\]
Indeed, if $N_0\le R+1$, then by deleting all zeros and, if
necessary, additional coordinates, we obtain a word with at
most one zero, which cannot contain $0110$ as a subsequence.
This contradicts our assumption. The same argument applied to
ones yields $N_1\ge R+2$.
Since $|z|=2R+4$, it follows that
\[
   N_0=N_1=R+2.
\]
Let
\[
   p_1<p_2<\cdots<p_{R+2}
\]
be the positions of the zeros in $z$. Fix two indices
$1\le i<j\le R+2$. Delete all zeros except those at positions
$p_i$ and $p_j$. Since this deletes exactly
\[
   (R+2)-2=R
\]
positions, the resulting word must contain $0110$ as a subsequence by our assumption.
However, the resulting word contains exactly two zeros, namely the zeros originally at positions $p_i$ and $p_j$.
Therefore these two zeros must serve as the first and last symbols of an occurrence of $0110$.
It follows that there are at least two ones between $p_i$ and $p_j$ in the original word $z$.
In particular, between every adjacent pair of zero positions
$p_i<p_{i+1}$, $1\le i\le R+1$, there are at least two ones.
Since these $R+1$ intervals are pairwise disjoint, we obtain
\[
   N_1\ge 2(R+1)=2R+2.
\]
This contradicts $N_1=R+2$.
Therefore, there exists a set $E\subseteq[2R+4]$ with
$|E|=R$ such that $z_{[2R+4]\setminus E}$ avoids $0110$ as a
subsequence. Hence there exists a codeword $c\in C_R$ such
that $c\preceq z$. Since $z$ was arbitrary, every word in
$\{0,1\}^{2R+4}$ is covered by $C_R$ under $R$ insertions.
Thus $C_R$ is an $R$-insertion-covering code.
\end{proof}

\begin{remark}
The construction above is deliberately small. The same obstruction can be lifted to arbitrarily large block lengths by appending unrestricted suffixes: for $n\ge R+4$, let $L=n-R-4$ and take
\[
   \{cv:\ c\in\{0,1\}^{R+4},\ 0110\not\preceq c,\ v\in\{0,1\}^L\}.
\]
The same argument applied to the first $2R+4$ positions gives insertion covering, while the word $0110\,0^L$ is not covered by deletions. Hence the non-equivalence is not a small-length artifact.
\end{remark}

Combining Theorems~\ref{thm:del-not-ins} and~\ref{thm:main-noneq} yields the promised non-equivalence.

\begin{corollary}\label{cor:noneq}
For every $R\ge1$, the properties of being an $R$-deletion-covering code and being an $R$-insertion-covering code are not equivalent. Neither property implies the other.
\end{corollary}

\section{A Lower Bound for \texorpdfstring{$R$}{R}-Insertion Covering Codes}\label{sec:insertion-lower}

We now turn to the lower bound concerning insertion covering codes, which is obtained by combining an inverse Bonferroni inequality with a tree structure on deletion positions.

Let $V_{n,R}$ be the set of $R$-tuples $(i_1,i_2,\ldots,i_R)$ with $1\le i_1<i_2<\cdots<i_R\le n+R$. Thus $|V_{n,R}|=\binom{n+R}{R}$. We define a tree $G_{n,R}$ on the vertex set $V_{n,R}$ recursively. When $R=1$, let $V_{n,1}=[n+1]$ and define the edge set by
\[
   E_{n,1}=\{\{i,i+1\}: i\in[n]\}.
\]
Suppose $E_{n,R-1}$ has been defined. We define $E_{n,R}$ by
\begin{align*}
E_{n,R}={}&\Bigl\{\{(i_1,\ldots,i_{R-1},i_R),
(i_1,\ldots,i_{R-1},i_R+1)\}:~1\le i_1<\cdots<i_{R-1}<i_R<n+R\Bigr\}\\
&\cup\Bigl\{\{(i_1,\ldots,i_{R-1},n+R),
(j_1,\ldots,j_{R-1},n+R)\}:~\{(i_1,\ldots,i_{R-1}),
(j_1,\ldots,j_{R-1})\}\in E_{n,R-1}\Bigr\}.
\end{align*}
Equivalently, one may describe $G_{n,R}$ as follows.
For any vertex
\[
   v=(i_1,i_2,\ldots,i_R)\neq (n+1,n+2,\ldots,n+R),
\]
let
\[
   t=\max\{s\in[R]: i_s<n+s\}.
\]
Then
\[
   v=(i_1,\ldots,i_t,n+t+1,\ldots,n+R),
\]
and $v$ is adjacent to
\[
   (i_1,\ldots,i_t+1,n+t+1,\ldots,n+R).
\]
Roughly speaking, vertices in $G_{n, R}$ are adjacent in lexicographical order. By checking that $G_{n, R}$ is connected and $|E_{n, R}|=|V_{n, R}|-1$, it is not difficult to see that $G_{n, R}$ is a tree.

Let $C\subseteq\Sigma_q^n$. For each
$v=(i_1,i_2,\ldots,i_R)\in V_{n,R}$, define
$C_v\subseteq\Sigma_q^{n+R}$ as follows: a word
$x\in\Sigma_q^{n+R}$ belongs to $C_v$ if, after deleting the
$i_1$-th, $i_2$-th, $\ldots$, $i_R$-th coordinates of $x$, the
resulting word belongs to $C$.

\begin{theorem}\label{thm:insert-lower}
If $C\subseteq\Sigma_q^n$ is an $R$-insertion covering code, then
\[
   |C|\ge \frac{q^n}{\binom{n+R}{R}-1}.
\]
Consequently,
\[
   I(q,n,R)\ge \frac{q^n}{\binom{n+R}{R}-1}.
\]
\end{theorem}

\begin{proof}
For simplicity, write
\[
   V=V_{n,R},\qquad G=G_{n,R},\qquad E=E_{n,R}.
\]
For an edge $e=\{u,v\}\in E$, write $C_e=C_u\cap C_v$. Applying
Lemma~\ref{lem:inverse-bonferroni}, we obtain
\begin{equation}\label{eq:inv-bonf-use}
   \sum_{e\in E}|C_e|
   \le
   \sum_{v\in V}|C_v|-
   \left|\bigcup_{v\in V}C_v\right|.
\end{equation}

For a fixed $v=(i_1,i_2,\ldots,i_R)\in V$, define a map
\[
   \pi_v:C_v\to C
\]
as follows: for $x\in C_v$, $\pi_v(x)$ is the word obtained from $x$ by deleting the $i_1$-th, $i_2$-th, $\ldots$, $i_R$-th coordinates.
By the definition of $C_v$, the map $\pi_v$ is well-defined.
Moreover, for every $y\in C$, there are exactly $q^R$ words $x\in C_v$ such that $\pi_v(x)=y$, since the deleted $R$ coordinates can be chosen arbitrarily from $\Sigma_q$. On the other hand, if $y_1\neq y_2\in C$, then the preimages of $y_1$ and $y_2$ are disjoint.
Thus
\[
   |C_v|=q^R|C|.
\]
Since $C$ is an $R$-insertion covering code, we also have
\[
   \bigcup_{v\in V}C_v=\Sigma_q^{n+R}.
\]
Consequently,
\[
   \left|\bigcup_{v\in V}C_v\right|=q^{n+R}.
\]
Substituting these identities into Inequality (\ref{eq:inv-bonf-use}) gives
\begin{equation}\label{eq:upper-edge-sum}
\sum_{e\in E}|C_e|
\le
\binom{n+R}{R}q^R|C|-q^{n+R}.
\end{equation}

We next lower bound $|C_e|$ for each edge $e\in E$. The same
calculation applies to every edge after relabeling the coordinates, so it
suffices to consider the representative edge
\[
   u=(1,2,\ldots,R-1,R),
   \qquad
   v=(1,2,\ldots,R-1,R+1).
\]
Then $x=(x_1,\ldots,x_{n+R})$ belongs to $C_u\cap C_v$ if and only if
\[
   (x_{R+1},x_{R+2},\ldots,x_{n+R})\in C
\]
and
\[
   (x_R,x_{R+2},x_{R+3},\ldots,x_{n+R})\in C.
\]
Let $1_C$ denote the indicator function of $C$. Then
\begin{align*}
|C_u\cap C_v|
&=
\sum_{x_1,\ldots,x_{n+R}\in\Sigma_q}
1_C(x_{R+1},x_{R+2},\ldots,x_{n+R})\cdot
1_C(x_R,x_{R+2},\ldots,x_{n+R})\\
&=q^{R-1}
\sum_{x_{R},\ldots,x_{n+R}\in\Sigma_q}
1_C(x_{R+1},x_{R+2},\ldots,x_{n+R}) \cdot
1_C(x_R,x_{R+2},\ldots,x_{n+R})\\
&=q^{R-1}\sum_{x_{R+2}, \ldots, x_{n+R}\in \Sigma_q}\left(\sum_{x_{R+1}\in\Sigma_q}{\rm{1}}_C(x_{R+1}, x_{R+2}, \ldots, x_{n+R})\cdot\sum_{x_R\in\Sigma_q}{\rm{1}}_C(x_{R}, x_{R+2}, x_{R+3}, \ldots, x_{n+R})\right)\\
&=
q^{R-1}
\sum_{x_{R+2},\ldots,x_{n+R}\in\Sigma_q}
\left(
   \sum_{a\in\Sigma_q}
   1_C(a,x_{R+2},\ldots,x_{n+R})
\right)^2.
\end{align*}
By Cauchy-Schwartz inequality,
\begin{align*}
&\sum_{x_{R+2},\ldots,x_{n+R}\in\Sigma_q}
\left(
   \sum_{a\in\Sigma_q}
   1_C(a,x_{R+2},\ldots,x_{n+R})
\right)^2\\
\geq&\frac{1}{q^{n-1}}\left(\sum_{x_{R+2}, \ldots, x_{n+R}\in \Sigma_q}\sum_{a\in\Sigma_q}1_C(a, x_{R+2}, \ldots, x_{n+R})\right)^2\\
=&
\frac{|C|^2}{q^{n-1}}.
\end{align*}
Therefore, for every edge $e\in E$,
\begin{equation}\label{eq:edge-lower}
   |C_e|\ge q^{R-n}|C|^2.
\end{equation}

Since $G$ is a tree,
\[
   |E|=\binom{n+R}{R}-1.
\]
Combining inequalities \eqref{eq:upper-edge-sum} and \eqref{eq:edge-lower}, we obtain
\[
   \left(\binom{n+R}{R}-1\right)q^{R-n}|C|^2
   \le
   \binom{n+R}{R}q^R|C|-q^{n+R}.
\]
Equivalently,
\[
   \left(\left(\binom{n+R}{R}-1\right)|C|-q^n\right)(q^n-|C|)\ge0.
\]
Thus
\[
   |C|\ge
   \frac{q^n}{\binom{n+R}{R}-1}.
\]
This proves the theorem.
\end{proof}

\section{An Upper Bound for \texorpdfstring{$R$}{R}-Deletion Covering Codes}\label{sec:deletion-upper}

We next consider deletion covering codes. In contrast to insertion balls, deletion balls are not regular, so we use an auxiliary hypergraph and apply Pippenger's theorem.


\begin{theorem}\label{thm:deletion-upper}
Let $n$ and $R$ be fixed. If $q\to\infty$, then
\[
   D(q,n,R)\le (1+o(1))\frac{q^{n-R}}{\binom nR}.
\]
Together with Inequality \eqref{eq:deletion-counting-bound}, this gives
\[
   D(q,n,R)=(1+o(1))\frac{q^{n-R}}{\binom nR}.
\]
\end{theorem}

\begin{proof}
Construct an auxiliary hypergraph $\mathcal H=(V,E)$ as follows.
Let $V$ be the set of all words in $\Sigma_q^{n-R}$ with distinct symbols, and let $E$ consist of the sets $\mathcal D_R(x)$, where $x\in\Sigma_q^n$ also has distinct symbols.
For such an $x$, all its $R$-deletion subsequences are distinct and still have distinct symbols. Hence $\mathcal H$ is $\binom nR$-uniform. Moreover,
\[
   |V|
   =
   q(q-1)\cdots(q-n+R+1)
   =
   (1+o(1))q^{n-R}.
\]

For each $y\in V$, the degree $d(y)$ equals the number of distinct-symbol words $x\in\Sigma_q^n$ containing $y$ as a subsequence. Thus
\[
   d(y)
   =
   \binom nR (q-n+R)(q-n+R-1)\cdots(q-n+1)
   =
   (1+o(1))\binom nR q^R.
\]
We next estimate the codegrees. Let $y,z\in V$ be distinct. If $y$ and $z$ use the same set of symbols, then, since they induce different orders on this set, no word with distinct symbols can contain both of them as subsequences. Hence $d(y,z)=0$ in this case. Otherwise, some symbol of $z$, say $z_j$, does not appear in $y$. To form a word $x$ containing both $y$ and $z$, first choose the positions of the symbols of $y$ in $x$, preserving their order; this gives at most $\binom nR$ choices. The symbol $z_j$ must occupy one of the remaining $R$ positions, and the other $R-1$ positions can be filled in at most $q^{R-1}$ ways. Therefore
\[
   d(y,z)\le \binom nR Rq^{R-1}=O(q^{R-1})=o(q^R),
\]
where $n$ and $R$ are fixed.

Thus, for all sufficiently large $q$, the hypergraph $\mathcal H$ satisfies the hypotheses of Theorem~\ref{thm:pippenger}. Hence it contains a cover $\mathcal E_1\subseteq E$ with
\[
   |\mathcal E_1|
   \le
   (1+o(1))\frac{|V|}{\binom nR}
   =
   (1+o(1))\frac{q^{n-R}}{\binom nR}.
\]
Let $C_1$ be the set of center words corresponding to the hyperedges in $\mathcal E_1$. Then
\[
   \bigcup_{x\in C_1}\mathcal D_R(x)\supseteq V.
\]

It remains to cover the words outside $V$, namely those with repeated symbols.
Choose $a\in \Sigma_q$ and let
\[
   C_2= \{ya^R: y\in\Sigma_q^{n-R}\setminus V\},
\]
then $C_2$ covers $\Sigma_q^{n-R}\setminus V$ and
\[
   |C_2|= |\Sigma_q^{n-R}\setminus V|= q^{n-R}-q(q-1)\cdots(q-n+R+1)=
   o(q^{n-R}).
\]
Consequently,
\[
   C=C_1\cup C_2
\]
is an $R$-deletion covering code, and
\[
   |C|
   \le
   |C_1|+|C_2|
   =
   (1+o(1))\frac{q^{n-R}}{\binom nR}.
\]
This completes the proof.
\end{proof}

\subsection{Single-deletion covering codes over large alphabets}

The preceding argument is non-constructive. For single deletion, we can give a more concrete construction by starting from differential VT codes and then completing the uncovered part.


\begin{definition}
For any $q\ge2$ and $a\in[0,qn-1]$, let
\[
   F_a=\{y\in \mathcal D_1(x):\ x\in\DVT_a(n,q)\}\subseteq\Sigma_q^{n-1}
\]
be the set of words covered by the differential VT code $\DVT_a(n,q)$ under one deletion, and let
\[
   \overline F_a=\Sigma_q^{n-1}\setminus F_a.
\]
\end{definition}

\begin{theorem}\label{thm:single-del-completion}
For any $q\ge2$ and $a\in[0,qn-1]$, define
\[
   C_a=\DVT_a(n,q)\cup \{y0: y\in\overline F_a\}.
\]
Then $C_a$ is a single-deletion covering code. Suppose that $n$ is a prime that does not divide $q$. Then there exists $a\in[0,nq-1]$ such that
\[
   |C_a|\le \frac{q^{n-1}+(q^{n-2}+1)(n-1)}{n}.
\]
Moreover, if $n$ is sufficiently large and $q>n^{1+c}$ for some constant $c>0$, then
\[
   |C_a|\le \frac{q^{n-1}}{n}(1+o(1)).
\]
\end{theorem}

\begin{proof}
We first verify the covering property. Since
\[
   \mathrm{DVT}_a(n,q)\subseteq \mathcal C_a,
\]
every word in $\mathcal F_a$ is covered by $\mathcal C_a$ under one deletion.
On the other hand, if $y\in\overline{\mathcal F}_a$, then by construction $y0\in\mathcal C_a$ and $y\in\mathcal D_1(y0)$.
Hence $\mathcal C_a$ is a single-deletion covering code.

It remains to bound the size of $\mathcal C_a$. Clearly,
\[
   |\mathcal C_a|
   \le
   |\mathrm{DVT}_a(n,q)|+|\overline{\mathcal F}_a|.
\]
When $n$ is a prime that does not divide $q$, Lemma~\ref{lem:dvt-size-prime} gives
\[
   |\mathrm{DVT}_a(n,q)|
   \le
   \frac{q^{n-1}-1}{n}+1.
\]
Thus it suffices to find a syndrome $a$ for which $\mathcal F_a$ is large.

Construct a bipartite graph
\[
   \mathcal G=(\Sigma_q^{n-1}\cup\Sigma_q^n,E),
\]
where $y\in\Sigma_q^{n-1}$ is adjacent to $x\in\Sigma_q^n$ if and only if
$y\in\mathcal D_1(x)$,
or equivalently, $x\in\mathcal I_1(y)$.
By Equation \eqref{eq:insball},
\[
   |E|=q^{n-1}\bigl((q-1)n+1\bigr).
\]

For each $a\in[0,qn-1]$, let $E_a$ be the set of edges in $E$ incident to the code $\mathrm{DVT}_a(n,q)$.
Since the codes
\[
   \mathrm{DVT}_a(n,q),\qquad a\in[0,qn-1],
\]
form a partition of $\Sigma_q^n$, the edge sets $E_a$ form a partition of $E$. Therefore, by the pigeonhole principle, there exists $a\in[0,qn-1]$ such that
\[
   |E_a|
   \ge
   \frac{|E|}{qn}
   =
   \frac{q^{n-2}\bigl((q-1)n+1\bigr)}{n}.
\]
By Lemma~\ref{lem:dvt-single-del}, the code $\mathrm{DVT}_a(n,q)$ is a single-deletion correcting code. Hence the deletion balls
\[
   \mathcal D_1(x),\qquad x\in\mathrm{DVT}_a(n,q),
\]
are pairwise disjoint. Consequently, each word in $\mathcal F_a$ is adjacent to exactly one codeword in $\mathrm{DVT}_a(n,q)$, and therefore
\[
   |\mathcal F_a|=|E_a|
   \ge
   \frac{q^{n-2}\bigl((q-1)n+1\bigr)}{n}.
\]
It follows that
\[
\begin{aligned}
   |\overline{\mathcal F}_a|
   = q^{n-1}-|\mathcal F_a|
   \le
   q^{n-1}
   -
   \frac{q^{n-2}\bigl((q-1)n+1\bigr)}{n}
   =
   \frac{(n-1)q^{n-2}}{n}.
\end{aligned}
\]
Combining this with the bound on $|\mathrm{DVT}_a(n,q)|$, we obtain
\[
\begin{aligned}
   |\mathcal C_a|
   &\le
   \frac{q^{n-1}-1}{n}+1
   +
   \frac{(n-1)q^{n-2}}{n}
   \le
   \frac{q^{n-1}+(q^{n-2}+1)(n-1)}{n}.
\end{aligned}
\]
This proves the desired bound.
\end{proof}

\section{Two-Burst-Deletion Covering Codes}\label{sec:burst}

We now turn to the construction of burst-deletion covering codes.
The key observation is that some known burst-deletion-correcting codes also possess the corresponding covering property.

\subsection{Binary codes}

We first show that binary differential VT codes give perfect covering codes for two-burst deletions.

\begin{theorem}\label{thm:binary-burst-cover}
For every $a\in[0,2n-1]$, the differential VT code $\DVT_a(n,2)$ is a binary two-burst-deletion covering code. Moreover, there exists $a\in[0,2n-1]$ such that
\[
   |\DVT_a(n,2)|\le \frac{2^{n-1}}{n}.
\]
\end{theorem}

\begin{proof}
Construct a bipartite graph
\[
   G=(\Sigma_2^{n-2}\cup \Sigma_2^n,E),
\]
where $y\in\Sigma_2^{n-2}$ is adjacent to $x\in\Sigma_2^n$ if and only if $y\in \BD_2(x)$, or equivalently, $x\in \BI_2(y)$. By Lemma~\ref{lem:burst-insertion-size} with $b=2$, each $y\in\Sigma_2^{n-2}$ has exactly $2n$ neighbors. Hence
\[
   |E|=2^{n-2}\cdot 2n=n2^{n-1}.
\]

For each $a\in[0,2n-1]$, let $E_a$ be the set of edges incident to $\DVT_a(n,2)$, and let $S_a\subseteq\Sigma_2^{n-2}$ be the set of words covered by $\DVT_a(n,2)$. Since by Lemma~\ref{lem:dvt-binary-burst} $\DVT_a(n,2)$ is a two-burst-deletion correcting code, the two-burst-deletion balls of distinct codewords in $\DVT_a(n,2)$ are disjoint. Therefore
\[
   |S_a|=|E_a|\le 2^{n-2}.
\]

On the other hand, the codes $\DVT_a(n,2)$, $a\in[0,2n-1]$, form a partition of $\Sigma_2^n$. Thus the edge sets $E_a$ form a partition of $E$, and so
\[
   \sum_{a=0}^{2n-1}|E_a|=|E|=n2^{n-1}.
\]
The average of the $2n$ numbers $|E_a|$ is $2^{n-2}$. Since each of them is at most $2^{n-2}$, we must have
\[
   |E_a|=2^{n-2}
\]
for every $a\in[0,2n-1]$. Hence $|S_a|=2^{n-2}$ for every $a$, which means that each $\DVT_a(n,2)$ covers the whole space $\Sigma_2^{n-2}$.

Finally, applying the pigeonhole principle to the partition of $\Sigma_2^n$ into the $2n$ differential VT codes gives a syndrome $a$ such that
\[
   |\DVT_a(n,2)|\le \frac{2^n}{2n}=\frac{2^{n-1}}{n}.
\]
This completes the proof.
\end{proof}

\begin{remark}
For each $a\in[0,2n-1]$, the code $\DVT_a(n,2)$ is simultaneously a binary two-burst-deletion correcting code and a binary two-burst-deletion covering code. In this sense, it is a binary perfect code for two-burst deletions.
\end{remark}

\subsection{Non-binary codes}

We next lift the binary construction to non-binary alphabets by separating the parity pattern from two checksum constraints over the quotient alphabet.

\begin{definition}
For an integer $m$, write $(m)_2$ for $m\pmod 2$. For a sequence $x=(x_1,\ldots,x_n)$, define
\[
  (x)_2=\big((x_1)_2,(x_2)_2,\ldots,(x_n)_2\big).
\]
\end{definition}

\begin{theorem}\label{thm:qary-burst-cover}
For any integers $q\ge2$, $a\in[0,2n-1]$, and
$c_0,c_1\in[0,\lfloor q/2\rfloor-1]$, define
\[
\begin{split}
C_{a,c_0,c_1}(n,q)=\{x\in\Sigma_q^n:\ &(x)_2\in\DVT_a(n,2),\\
&\sum_{i=1}^{\lceil n/2\rceil}
\left\lfloor\frac{x_{2i-1}}2\right\rfloor
\equiv c_0\pmod{\lfloor q/2\rfloor},\\
&\sum_{i=1}^{\lfloor n/2\rfloor}
\left\lfloor\frac{x_{2i}}2\right\rfloor
\equiv c_1\pmod{\lfloor q/2\rfloor}\}.
\end{split}
\]
Then $C_{a,c_0,c_1}(n,q)$ is a $q$-ary two-burst-deletion covering code. Moreover, there exists a choice of parameters such that
\[
   |C_{a,c_0,c_1}(n,q)|\le \frac{q^n}{2n\lfloor q/2\rfloor^2}.
\]
\end{theorem}

\begin{proof}
Let $y=(y_1,\ldots,y_{n-2})\in\Sigma_q^{n-2}$. Since by Theorem~\ref{thm:binary-burst-cover} $\DVT_a(n,2)$ is a binary two-burst-deletion covering code, there exist an index $i\in[1,n-1]$ and bits $d_0,d_1\in\{0,1\}$ such that
\[
   (y_1,\ldots,y_{i-1},d_0,d_1,y_i,\ldots,y_{n-2})_2\in\DVT_a(n,2).
\]

Define
\[
   s_j\equiv c_j-
   \sum_k\left\lfloor\frac{y_{2k-1+j}}2\right\rfloor
   \pmod{\lfloor q/2\rfloor},
   \qquad j=0,1,
\]
where the sum is over all valid indices.

If $i$ is odd, set
\[
   x=(y_1,\ldots,y_{i-1},2s_0+d_0,2s_1+d_1,y_i,\ldots,y_{n-2}).
\]
Then the inserted symbols occupy one odd position and one even position, respectively. Hence the parity pattern of $x$ belongs to $\DVT_a(n,2)$, and the two checksum constraints are satisfied by the choices of $s_0$ and $s_1$. Therefore
\[
   x\in C_{a,c_0,c_1}(n,q)
   \quad\text{and}\quad
   y\in\BD_2(x).
\]

If $i$ is even, we instead set
\[
   x=(y_1,\ldots,y_{i-1},2s_1+d_0,2s_0+d_1,y_i,\ldots,y_{n-2}).
\]
In this case the first inserted symbol occupies an even position and the second occupies an odd position, so the same argument shows that $x\in C_{a,c_0,c_1}(n,q)$ and $y\in\BD_2(x)$.

Thus every $y\in\Sigma_q^{n-2}$ is covered by some codeword in $C_{a,c_0,c_1}(n,q)$, and hence $C_{a,c_0,c_1}(n,q)$ is a $q$-ary two-burst-deletion covering code.

Finally, as $a,c_0,c_1$ vary, the sets $C_{a,c_0,c_1}(n,q)$ form a partition of $\Sigma_q^n$. Since there are $2n\lfloor q/2\rfloor^2$ choices of parameters, the pigeonhole principle gives a choice satisfying
\[
   |C_{a,c_0,c_1}(n,q)|
   \le
   \frac{q^n}{2n\lfloor q/2\rfloor^2}.
\]
This completes the proof.
\end{proof}

\section{Conclusions}\label{sec:conclusion}

We conclude by summarizing the main results and pointing out several directions for future work.
In this paper, we study covering codes under insertion and deletion operations. We first show that insertion covering and deletion covering are not equivalent. We then develop bounds and constructions for insertion and deletion covering codes, with particular emphasis on the large-alphabet regime.
For insertion covering codes, we established a new lower bound that improves the standard sphere-covering bound in the large-alphabet regime.
For deletion covering codes, we proved that the elementary counting lower bound is asymptotically tight when the alphabet size tends to infinity.
We also gave a construction of non-binary single-deletion covering codes based on differential VT codes, which attains the asymptotically optimal size under suitable parameter choices.
Finally, we showed that certain burst-deletion-correcting codes also have the corresponding covering property, and used this observation to construct two-burst-deletion covering codes.

The following questions appear to be natural in view of the constructions and bounds developed in this paper.
\begin{itemize}[leftmargin=2em]
\item Although the asymptotic upper bound for general $R$-deletion covering codes is optimal in the large-alphabet regime, its proof is non-constructive. It would be interesting to find explicit constructions matching this bound for general $R$.

\item For insertion covering codes, there remains a gap between the known lower and upper bounds, even in the large-alphabet regime. Determining the exact asymptotic behavior of $I(q,n,R)$ for fixed $R$ remains an important open problem.

\item Recent semidefinite programming techniques have led to improved lower bounds for covering codes in the Hamming metric \cite{2025arXiv250401932G}. It would be interesting to see whether such techniques can be adapted to covering codes under insertions or deletions.

\item For non-binary burst-deletion covering codes, improving the redundancy of explicit constructions and developing efficient encoding algorithms remain challenging problems.

\item It would also be worthwhile to investigate covering codes under more general synchronization-error models, such as mixed insertion-deletion errors and multiple-burst deletion channels.
\end{itemize}

\section*{Acknowledgment}
The first author would like to thank Wenjun Yu and Tingting Chen for many helpful and valuable comments.

\bibliographystyle{abbrv}
\bibliography{ref}
\end{document}